\documentclass[leqno]{article} 

\usepackage{amssymb,amsthm,amsmath,amsfonts}
\usepackage{latexsym}
\usepackage{graphics} 
\usepackage{epsfig}
\title{TITLE}
\author{AUTHOR}
\date{\today}

\theoremstyle{plain}
\newtheorem{thm}{Theorem}
\newtheorem{lem}{Lemma}
\theoremstyle{definition}
\theoremstyle{remark}

\newcommand{\BM}{\bold{(BM)}}
\newcommand{\NIG}{\bold{(NIG)}}

\begin{document}

\renewcommand{\include}[1]{}
\title{Asymptotic behavior of prices of path dependent options}
\author{
	Yuji Hishida (Mizuho Securities Co., Ltd),
 	Kenji Yasutomi (Ritsumeikan University) \\
	e-mail: yuji.hishida@mizuho-sc.com
	        yasutomi@se.ritsumei.ac.jp
}
\maketitle
\begin{abstract}
In this paper, we give a numerical method for pricing long maturity, 
path dependent options by using the Markov property for each 
underlying asset.
This enables us to approximate a path dependent option 
by using some kinds of plain vanillas. 
We give some examples whose underlying assets behave as some popular Levy 
processes. Moreover, we give some payoffs and functions used to
approximate them.

\end{abstract}

\section{Introduction}
In this paper, we give a numerical method for pricing 
some path dependent options by using the Markov property
for each underlying asset process.
Path dependent options are options whose payoff
at maturity depend on the past history of 
the underlying asset as well as the price at maturity.
Asian options, Lookback options and barrier options are
thier well-known examples.
Some of them as the above are difficult to calculate
analytically and numerically, and there have been numerous
studies on how to do this.

One numerical method of pricing path dependent options
uses Monte Carlo simulation.
It has an advantage in that with it, we can simulate the 
expectation of the payoff
without detailed discussion of the payoff type, 
but only with the distribution of the underlying asset.
N. Hilber, N. Reich, C. Schwab and C. Winter \cite{NNCC}
or Cont, R and  Tankov, P \cite{ContTankov}
give the Monte Carlo method when
market models are extended to Levy processes.
They mention that most of Levy processes
can be simulated approximately.
As seen from them,
the Monte Carlo method has the applicability
for many payoff types and many underlying asset models. 
However, the longer a maturity time is, the more
the calculation time is consumed.

Several analytical methods of path dependent options
have been also developed.
An analytical formula for standard barrier options is 
given in papers by Merton \cite{Merton} or Reiner and Rubinstein 
\cite{ReinerRubin} when the underlying asset behaves as 
the geometric Brownian motion. 
An analytical formula 
for partial barrier options is also given by
Armstrong \cite{Armstrong} or Heynen and Kat \cite{HeynenKat}
under the same conditions. 
When the underlying asset is more general such as 
a Levy process, it is not easy to find the exact
value of them. There are some studies about
this problem, for example,
Kudryavstev and Levendorski \cite{KudLev}
give the Fast Wiener-Hopf factorization method.
Asian options are difficult to calculate 
analytically. 
Geman and Yor \cite{GY} give 
a semi-analytical formula by using the Laplace transformation.
Linetsky \cite{Linetsky1} uses a spectral expansion approach, 
and Benhamou \cite{Benhamou} uses a convolution method 
for pricing them.
Albrecher and Predota \cite{AlbPredota} 
approximate the arithmetic
option price based on the moments of the average.

Most approaches of these methods depend on the payoff functions.
In contrast, a method using the Malliavin-Watanabe calculus given in 
Kumitomo and Takahashi \cite{KT0}
is valid for some path dependent options.
Bermin \cite{Hans1} and \cite{Hans2} show that the Malliavin calculus
approach can be applied for any square integrable payoff.

Our method is an analytical approximation approach.
In \cite{HY}, we give the asymptotic behavior of the 
prices of an Asian option about maturity time $T$
as the underlying asset behaves as the
geometric Brownian motion. 
When the maturity time $T \to \infty$, 
the expectation of the payoff for an Asian option can be 
approximated by that of the payoff for 
some linear combinations of plain vanillas; i.e.,
there exists a non-zero constant, $D$, such that 
when $T \to \infty$,
\begin{eqnarray*}
E[(\frac{1}{N}\sum_{i=1}^N S_{T-\tau_i} - K)^{+}]
\simeq
\frac{1}{N}\sum_{i=1}^N E[(S_{T-\tau_i} - K)^{+}] +
D \alpha(T)
\end{eqnarray*}
where the underlying asset is defined by 
$S_t = S_0\exp(\sigma W_t + (r-\frac{1}{2}\sigma^2)t )$,
$K$ is a strike price, $\{\tau_i\}$ is fixed time, 
$W=\{W_t\}_{0\leq t \leq T}$ is a one-dimensional 
Brownian motion, and $\alpha$ is defined by
\begin{eqnarray}
	\alpha(T) = \frac{1}{\sqrt{T}}\exp
	\left(- \frac{1}{2\sigma^2}(r-\frac{1}{2}\sigma^2) T
	\right). \label{decayBM}
\end{eqnarray}
In this paper, we generalize the above idea
to an almost universal situation; i.e.,
to some payoff functions and a class of
the underlying asset process.

To summarize our result of this paper,
let the underlying asset behave as positive 
Markov process $S_t$. 
We represent an expectation of a path dependent option with maturity 
time $T$ as
\begin{eqnarray*}
	A^{*}(T) = E[g( (S_s)_{ T-\tau \le s < T } )]
\end{eqnarray*}
where $g$ is a function from path space
which represents the payoff of the path dependent option.
Roughly speaking, if $S_t$ decays uniformly in $\mathbb{R}$; 
i.e., there exists $\alpha^{*} : [0,\infty) \to (0,\infty)$
such that
$
	\frac{P( dS_T )}{\alpha^{*}(T)} 
$ 
``convergence'' to non-trivial measure as $T\to\infty$,
and also if $g$ satisfys ``good'' conditions, 
there exists
$\widetilde A^{*}(T)$ and a constant, $C$,
such that the error term
between $A^{*}$ and $\widetilde A^{*}$ can be estimated 
when $T\to\infty$; i.e.,
\begin{equation}
	 \frac{ A^{*}(T) - \widetilde A^{*}(T) }{ \alpha^{*}(T) } \to C. 
\label{approx1}
\end{equation}
In practical sense, if $\widetilde A^{*}(T)$ is either easy to calculate
or is quoted in the market,
we easily obtain a numerical approximation of $A^{*}(T)$ from
\begin{equation*}
	A^{*}(T) \simeq  \widetilde A^{*}(T) + C\alpha^{*}(T).
\end{equation*}
This means that we can approximate a path dependent option
by using $\widetilde A^{*}$ and an error term.

Although this asymptotic approach is only valid 
for calculating the value of long maturity options,
it is beneficial because
\begin{enumerate}
\item We can get the value of $A^{*}(T)$  instantaneously
	since  $\alpha^{*}$ is written as elementary functions
	in most cases. \label{enu1}
\item We can apply this method to many path dependent options. \label{enu2}
\item We can apply this method to a large class of the underlying asset. \label{enu3}
\end{enumerate}
As an example of the first case,
when $S_t$ is geometric Brownian motion, $\alpha^{*}$ is 
given by (\ref{decayBM}) and written as elementary functions.
As an example of the second case,
we give the result in the case of an Asian option, 
the Lookback option, and the barrier option.
As an example of the third case, we give the result in case 
of a geometric Levy process, particularly for Brownian motion $\BM$
and for the Normal Inverse Gaussian process $\NIG$.
In our method, the longer a maturity time is,
the better the computation accuracy is.
That is because our method is based on the
asymptotic behavior of price as $T$ tends to $\infty$.
This is a notable result since most methods,
like Monte Carlo simulation, 
are effective when the variance of the underlying asset is small; i.e.,
the maturity time is short.

This paper is organized as follows. In Section~\ref{generalsection}, 
we present a theorem which gives the principle of our method
and give the short proof.
In Section~\ref{alphasection}, we give some exmaples of $\alpha^{*}$ 
when the underlying asset behaves as some popular Levy processes, 
and we see that these $\alpha^{*}$ are written by elementary functions
in these cases.
In Section~\ref{payoffsection}, we give two theorems which
help us how to find the approximation function $\widetilde A^{*}(T)$ 
when given $A^{*}(T)$. Also, we show that some path dependent options,
 such as an Asian option, a Lookback option and a barrier option, can be
applied these theorems.
In Section~\ref{proofsection}, theorems of Section~\ref{payoffsection}
are proved, and in Section~\ref{appendixsection}, we present some properties
of $\BM$ and $\NIG$ used in the proof.

\section{Generalized Principle} \label{generalsection}
We set the notations again.
Let an underlying asset process, $S_s$ be a positive 
Markov process.
For a fixed time $\tau > 0$, 
we denote the set of all c\'{a}dl\'{a}g functions with domain $[0,\tau]$ 
by $D([0,\tau])$. We also define an expectation of a path dependent option
whose maturity time is $T$ and whose monitoring period
is $[T-\tau, T]$ by
\begin{eqnarray*}
A^{*}(T) := E[g((S_s)_{T-\tau \le s \le T})]
\end{eqnarray*}
where $g$ is a function from $D([0,\tau])$ into $\mathbb{R}$.
For given $A^{*}(T)$, we also define $\widetilde A^{*}(T)$ by
\begin{eqnarray*}
\widetilde A^{*}(T) := E[\widetilde g((S_s)_{T-\tau \le s \le T})]
\end{eqnarray*}
where $\widetilde g$ is also a function from $D([0,\tau])$ into $\mathbb{R}$.
In this paper, we treat this $\widetilde g$ 
as the approximation function of the path dependent option.
%
For the sake of simplicity,
we use the notation $t:=T-\tau$,
\begin{eqnarray*}
&&A(t) := A^{*}(T) = E[g((S_s)_{t \le s \le t + \tau})], \\
&&\widetilde A(t) := \widetilde A^{*}(T)
  = E[\widetilde g((S_s)_{t \le s \le t + \tau})],
\end{eqnarray*}
and we see the asymptotic behavior as $t \to \infty$
instead of the asymptotic behavior as $T \to \infty$.
Under these conditions,
we derive the following theorem. 
\begin{thm}\label{thm_GP}
We assume two assumptions, \ref{conv} and \ref{integ}
about $(g,\widetilde g,S)$;
\renewcommand{\theenumi}{A\arabic{enumi}}
\begin{enumerate}
\item \label{conv}
	Process $S$ decays uniformly; that is, there exist
	measures $\nu$ and $\bar{\nu}$ and 
	a funcion, $\alpha : [0,\infty) \to (0,\infty)$, such that 
	\begin{eqnarray*}
		&& \nu_t \gg \bar{\nu}, \quad \nu \gg \bar{\nu}, \\
		&& \frac{d\nu_t}{d\bar{\nu}}
			\to \frac{d\nu}{d\bar{\nu}} \quad (t \to \infty)
	\end{eqnarray*}
	where 
	$\nu_t(M) := \frac{P(S_t \in M)}{\alpha(t)}$.
	We use the notaion $\gg$ to mean absolutely continuous, and
	$\frac{d\nu_t}{d\bar{\nu}}$ it the  Radon Nycodim derivative.
\item \label{integ}
	\begin{eqnarray*}
		 \int |\epsilon(x)| \sup_{t} \frac{d\nu_t}{d\bar{\nu}}(x) 
		\, \bar{\nu}(dx) < \infty
	\end{eqnarray*}
	where
	\begin{eqnarray*}
		 \epsilon(x) := E[g( (S_s^x)_{ 0 \le s \le \tau} ) 
    	-\widetilde g( (S_s^x)_{ 0\le s \le \tau} )]
	\end{eqnarray*}
	and $S_s^x$ is a process starting at $x$; i.e., 
        $S_s^x := \frac{x}{S_0} S_s$.
\end{enumerate}
	Then it follows that when $t \to \infty$,
	\begin{eqnarray}\label{approx2}
		\frac{ A(t) - \widetilde A(t) }{ \alpha(t) } 
		\to \int \epsilon(x) \, \nu(dx).
	\end{eqnarray}
\end{thm}
We remark that, for $\alpha^{*}(T) := \alpha(t+\tau)$, since
\begin{eqnarray*}
 \frac{ A^{*}(T) - \widetilde A^{*}(T) }{ \alpha^{*}(T) }
 =
 \frac{ A(t) - \widetilde A(t) }{ \alpha(t) },
\end{eqnarray*}
(\ref{approx1}) means (\ref{approx2}).
\begin{proof}[Proof of the theorem]
	From the Markov property of $S$, we obtain
	\begin{eqnarray*}
		A(t) - \widetilde A(t)
		&=&
		E[g( (S_s)_{ t \le s \le t+ \tau} ) 
			-\widetilde g( (S_s)_{ t \le s \le t + \tau} )] \\
		&=&
		E[E[g( (S_s)_{ t \le s \le t+ \tau} ) 
			-\widetilde g( (S_s)_{ t \le s \le + \tau} ) | S_t ] ] \\
		&=& E[\epsilon(S_t)].
	\end{eqnarray*}
	Since 
	$\int |\epsilon(x)| \sup_{t} \frac{d\nu_t}{d\bar{\nu}}(x) 
	\, \bar{\nu}(dx) < \infty$,
	the Lebesgue convergence theorem implies that
	\begin{eqnarray*}
		\frac{ A(t) - \widetilde A(t) }{ \alpha(t) }
		&=&
		\int \epsilon(x) \frac{d\nu_t}{d\bar{\nu}}(x) \, \bar{\nu}(dx) \\
		&\to&
		\int \epsilon(x) \frac{d\nu}{d\bar{\nu}}(x) \, \bar{\nu}(dx) \\
		&=&
		\int \epsilon(x) \, \nu(dx).
	\end{eqnarray*}
\end{proof}
\section{Asymptotic Order of each model} \label{alphasection}
In this section, we consider \ref{conv}
in Theorem~\ref{thm_GP}. 
The existence of $\alpha(t)$ 
depends only on each process, $S$.
It is clear that $\alpha$ is unique in the sense of order; i.e.,
$\alpha(t)\sim \alpha^{\prime}(t)$ as $t\to\infty$
if $\alpha$ and $\alpha^{\prime}$ satisfy \ref{conv} 
for the same process.
We give some examples of $S_t$ satisfying \ref{conv}
represented as $S_t = e^{Z_t}$,
where $Z_t$ is a popular Levy process.
In the following argument, we regard measure $\bar{\nu}$ 
as a Lebesgue measure.

For the sake of simplicity, we discuss $Z$ instead of $S$.
Let $\hat\nu_t$ be a modified distribution of $Z_t$; i.e.,
$ \hat\nu_t(A) := \frac{P( Z_t\in A )}{\alpha(t)} $.
Then, since $S_t = e^{Z_t}$, we have that
$$
	\frac{d\hat\nu_t}{d\bar\nu}(z) 
	= \frac{d}{dz} \frac{P(Z_t\le z)}{\alpha(t)}
	= \frac{d}{dz} \frac{P(e^{Z_t}\le e^z)}{\alpha(t)}
	= \frac{de^z}{dz} \frac{d}{de^y} \frac{P(S_t\le e^z)}{\alpha(t)}
	= e^z \frac{d\nu_t}{d\bar\nu}(e^z).
$$
Thus, by substituting $e^y=x$, the condition \ref{conv} is replaced by
a condition about $Z$;
\begin{lem}
$\hat\nu_t \gg \bar{\nu}$, $\hat\nu \gg \bar{\nu}$, and 
\begin{equation} \label{conv_Z}
	\frac{d\hat\nu_t}{d\bar\nu}(z) \to \frac{d\hat\nu}{d\bar\nu}(z)
\end{equation}
for any $z$ imply \ref{conv}
for $\frac{d\nu}{d\bar\nu}(x) := \frac{1}{x} \frac{d\hat\nu}{d\bar\nu}(\log x)$.
\end{lem}
\subsection{Brownian Model}
Let us define $Z_t = z_0 + \sigma W_t + \mu t$ where 
$z_0$, $\mu \in \mathbb{R}$, $\sigma > 0$, 
and $\{W_t\}_{0 \leq t < \infty}$ is a one-dimensional Brownian motion. 
Then
\begin{eqnarray*}
	\alpha(t)
	&=& 
	\frac{1}{\sqrt{t}}\exp(-\frac{\mu^2 t}{2\sigma^2}) \\
	\mbox{and} \quad 
	\frac{d\hat\nu}{d\bar\nu}(z)
	&=& \frac{1}{\sqrt{2\pi\sigma^2}}\exp({\frac{\mu z}{\sigma^2}}) 
\end{eqnarray*}
satisfy (\ref{conv_Z}).

\subsection{Normal Inverse Gaussian Model}
Let us define $Z_t = z_0 + W_{\mbox{\rm IG}(t)} 
+ \theta \mbox{\rm IG}(t)+ b t$
where $z_0$, $b \in \mathbb{R}$, 
and we denote the inverse Gaussian subordinator, $\mbox{\rm IG}$, by 
$\mbox{\rm IG}(t) = \inf\{s>0;B_s + \mu s > \delta t\}$, 
where $\mu \in \mathbb{R}$, $\delta > 0$, 
and $B$ is another Brownian motion independent of $W$. 
This process is called normal inverse Gaussian, see \cite{DA}. 
Then
\begin{eqnarray*}
	\alpha(t) 
	&=& t^{-\frac{1}{2}}
        \exp\Big(  t\Big( \mu\delta - \theta b - \sqrt{(b^2+\delta^2)(\theta^2+\mu^2)}
              \Big) \Big) \\
	\mbox{and} \quad
	\frac{d\hat\nu}{d\bar{\nu}}(z)
	 &=& 
		\frac{\delta}{ \sqrt{2\pi(b^2+\delta^2) } }
		\left( 
			\frac{\theta^2 + \mu^2}{ b^2+\delta^2 } 
		\right)^{\frac{1}{4}}
		e^{ \left(
				\theta
				+
				b
				\sqrt{ \frac{
					\theta^2 + \mu^2
				}{
					b^2 + \delta^2
				}}
			\right)
			(z-z_0)
		}
\end{eqnarray*}
satisfy (\ref{conv_Z}).
\subsection{Variance Gamma Model}
Let us define $Z_t = z_0 + \sigma W_{\gamma(t)} + \theta \gamma(t) + mt$, 
where $z_0 \geq 0$, $\theta$, $m \in \mathbb{R}$, and we denote 
the gamma process by $\gamma(t)$ with variance rate $\lambda$, see \cite{VG}. 
Then
\begin{eqnarray*}
\alpha(t) &=&
 	t^{-\frac{1}{2}}
	\left(
		\frac{1+\eta}{ 2+\nu\theta^2/\sigma^2}
	\right)^{\frac{t}{\nu}}
	e^{ ( \frac{1-\eta}{\lambda} - \frac{\theta m}{\sigma^2}) t } \\
	\mbox{and} \quad
   \frac{d\hat\nu}{d\bar\nu}(z) &=&
	\frac{1}{2}
	\sqrt{ 
		\frac{ 2+\lambda\theta^2/\sigma^2 }{ 2\pi \eta(1+\eta) \sigma^2 }
	}
 	\exp \left( 
		( \frac{\theta}{\sigma^2} + \frac{\eta-1}{m \lambda} )(z-z_0)
	\right) 
\end{eqnarray*}
satisfy (\ref{conv_Z})
where 
$
	\eta :=\sqrt{ 1 + \frac{ m^2 \lambda^2 }{\sigma^2} (\frac{2}{\lambda} + \frac{\theta^2}{\sigma^2} )}
$.
\section{Path dependent Payoff and its Approximation function} \label{payoffsection}
In this section, we consider \ref{integ}
in Theorem~\ref{thm_GP}.
To satisfy \ref{integ} for $(g,\widetilde g, Z)$,
we need to select a ``good" $\widetilde g$ 
for a given path dependent payoff, $g$.
The following theorems give classes of $(g,\widetilde g)$,
which satisfy \ref{integ}.
The first class includes, as examples, Asian options and a Lookback option.
\begin{thm}\label{payoffthm}
	Let $h:D( [0,\tau] ) \to \mathbb{R}$ satisfy
	\begin{eqnarray}
		&&\inf_{0\le s \le\tau} w(s) \le h(w) 
		\le \sup_{0\le s \le\tau} w(s), \ \label{hbounded} \\
		&& h(aw) = ah(w), \label{schalar}
	\end{eqnarray}
	for any $w \in D( [0,\tau] )$ and $a \in \mathbb{R}$,
	and let $(g,\widetilde g)$ represent
	\begin{eqnarray*}
		g(w) = (h(w)-K)^{+}, \quad
		\widetilde g(w) = C(e^{w(0)} - K/C )^{+},
	\end{eqnarray*}
	where $C := E[h( (S_t)_{0\le t \le \tau} )]$. 
	Then \ref{integ} holds
	in cases of $\BM$ and $\NIG$ with the parameter condition,
	$$
		\frac{ 2b }{ \delta + \sqrt{b^2+\delta^2} } < 1.
	$$
\end{thm}

The second class includes a barrier option. 
It also includes the Asian options with another $\widetilde g$.
\begin{thm}\label{payoffthm2}
	Let $(g,\widetilde g)$ satisfy there exist $M_U$, $M_L$ and $M$ such that
	\begin{eqnarray}
		\inf_{0\le s \le\tau} w(s) > M_{L} 
		\Rightarrow g(w) = \widetilde g(w), \label{payoffthm2a1} \\
		\sup_{0\le s \le\tau} w(s) < M_{U}
		\Rightarrow g(w) = \widetilde g(w), \label{payoffthm2a2} \\
		|g(w) - \widetilde g(w)| \le \sup_{0 \le s \le\tau} w(s) + M. \label{payoffthm2a3} 
	\end{eqnarray}
	Then \ref{integ} holds
	in cases of $\BM$ and $\NIG$ with the parameter condition,
	$$
		\frac{ 2b }{ \delta + \sqrt{b^2+\delta^2} } < 1.
	$$
\end{thm}

We note that the classes given by
Theorem~\ref{payoffthm} and Theorem~\ref{payoffthm2}
are just examples of class satisfying \ref{integ}
and that there may exist many other classes.
Especially, for given payoff $g$, the $\widetilde g$ is not unique. 

\subsection{Discrete Asian option}
The payoff for a discrete Asian option whose strike price is $K$ 
and whose maturity time is $T$ is defined by
$$
\left( \frac{1}{N}\sum_{i=1}^{N}S_{T-\tau_i} - K \right)^{+}
$$
where $0 \leq \tau_1 \leq \cdots\leq \tau_N = \tau$. 
In this case, let $h$ be defined by
\begin{eqnarray*}
	h(w) := \frac{1}{N}\sum_{i=1}^{N} w(\tau - \tau_i).
\end{eqnarray*}
We can then express the option as
\begin{eqnarray*}
	g((S_s)_{T-\tau\le s \le T}) 
	&=&
	\left( 
		h((S_s)_{T-\tau\le s \le T}) - K
	\right)^{+} \\
	&=&
	\left( \frac{1}{N}\sum_{i=1}^{N}S_{T-\tau_i} - K \right)^{+}.
\end{eqnarray*} 
Then it is easy to see that
$h$ satisfies the conditions of Theorem~\ref{payoffthm}.

To apply Theorem~\ref{payoffthm2},
let $(g,\widetilde g)$ be defined by
\begin{eqnarray*}
        g(w) &:=& \left( \frac{1}{N}\sum_{i=1}^{N} w(\tau - \tau_i) - K \right)^{+},\\
	\widetilde g(w) &:=&  \frac{1}{N}\sum_{i=1}^{N} 
        \left( w(\tau - \tau_i) - K_i \right)^{+}
\end{eqnarray*} 
where $\frac{1}{N}\sum_{i=1}^{N} K_i = K$.
Then $(g,\widetilde g)$ satisfys the conditions of Theorem~\ref{payoffthm2}
with respect to $M_U:= \min_i K_i$ and $M_L:= \max_i K_i$.

\subsection{Integral Asian Option}
Similar to the payoff for a discrete Asian option, 
we define that for an Asian option as 
$$
	\left( \frac{1}{\tau}\int_{T-\tau}^{T} S_s \ ds  - K \right)^{+}.
$$
In this case, let $h$ be defined by
\begin{eqnarray*}
	h(w) := \frac{1}{\tau}\int_{0}^{\tau} w(s) ds.
\end{eqnarray*} 
Then $h$ satisfies the conditions of Theorem~\ref{payoffthm}.

To apply Theorem~\ref{payoffthm2},
let $(g,\widetilde g)$ be defined by
\begin{eqnarray*}
        g(w) &:=& 
        \left( \frac{1}{\tau}\int_{0}^{\tau} w( s )\, ds - K \right)^{+} ,\\
	\widetilde g(w) &:=& 
        \frac{1}{\tau}\int_{0}^{\tau} \left( w( s ) - K_s \right)^{+}\,ds
\end{eqnarray*} 
where $\frac{1}{\tau}\int_{0}^{\tau}K_s \,ds = K$.
Then $(g,\widetilde g)$ satisfys the conditions of Theorem~\ref{payoffthm2}
with respect to $M_U:= \inf_s K_s$ and $M_L:= \sup_s K_s$.

\subsection{Discrete Lookback Option}
The payoff for a discrete Lookback option whose strike price is $K$ 
and whose maturity time is $T$ is defined by 
$$ 
	\left( \max_{i=1, \ldots, N} S_{T-\tau_i} - K \right)^{+}
$$
where $0 \leq \tau_1 \leq \cdots\leq \tau_N = \tau$. 
In this case, let $h$ defined by
\begin{eqnarray*}
	h(w) := \max_{i=1, \ldots ,N} w(\tau_i).
\end{eqnarray*}
Then $h$ satisfies the conditions of Theorem~\ref{payoffthm}.

\subsection{Partial Barrier Option}
The payoff for a knock-in partial barrier option 
with the strike price $K$,
the barrier level $L$,
the maturity time $T$, 
and monitoring period $[T-\tau,T]$
is defined by
\begin{eqnarray*}
	( S_T  - K )^{+} 1_{ \{ \displaystyle \inf_{T-\tau \le t \le T} S_t \ge L \} }
\end{eqnarray*} 
where $L>K$.
Let $(g,\widetilde g)$ be defined by
\begin{eqnarray*}
	g(w)            &:=&  ( w(\tau) - K )^{+} 1_{ \{ \inf w \ge L \} },\\
	\widetilde g(w) &:=&  ( w(\tau) - L )^{+}.
\end{eqnarray*}
Then $(g,\widetilde g)$ satisfies the conditions of Theorem~\ref{payoffthm2}
with respect to $M_U := K$ and $M_L := L$.
\section{Proof of Theorem~\ref{payoffthm} and Theorem~\ref{payoffthm2}} 
\label{proofsection}
To prove theorems,
we use Lebesgue integrability.
By substituting $e^z=x$,
the condition \ref{integ} is replaced by
a condition about $Z$;
\begin{lem}
	\begin{eqnarray}\label{integ_mod}
		 \int_{-\infty}^{\infty} |\hat \epsilon(z)| \sup_{t} \frac{d\hat \nu_t}{d\bar{\nu}}(z) 
		\, \bar{\nu}(dz) < \infty
	\end{eqnarray}
	implies \ref{integ}
	for $\hat \epsilon(z) := \epsilon(e^z)$.
\end{lem}

To see (\ref{integ_mod}), 
it is enough to demonstrate that the integrand vanishes fast 
when $|z| \to \infty$.
At first, we start with descriptions of the bound of 
$\sup_{t} \frac{d\hat\nu_t}{d\bar\nu}(x)$. 
Note that this density function, $\sup_{t}  \frac{d\nu_t}{d\bar\nu}(x)$, 
depends only on the process $Z$ but not payoff and its approximation function, 
$(g,\widetilde g)$;
\begin{eqnarray}
	\BM && \sup_{t}  \frac{d\hat\nu_t}{d\bar\nu}(z) 
	\le C_{\BM}
		\exp\left( \frac{ \mu }{\sigma^2}z \right),
		\label{BMsup}\\ 
	\NIG && \sup_{t}  \frac{d\hat\nu_t}{d\bar\nu}(z) 
		\le C_{\NIG}
		\exp\left( 
			\left( 
				\theta + b
				\frac{ 2\sqrt{ \theta^2 + \mu^2 } }{ \delta + \sqrt{b^2+\delta^2} }
			\right)
			z
		\right)
		\label{NIGsup} 
\end{eqnarray}
where $C_{\BM}$ and  $C_{\NIG}$ are constants.

We next give a lemma about the order for 
$$
	\hat \epsilon(z)
	=
	E[g( (e^{Z_s^z})_{ 0 \le s \le \tau } ) 
		-\widetilde g( (e^{Z_s^z})_{ 0 \le s \le \tau } )]
$$ 
to prove Theorem~\ref{payoffthm} 
where $Z_s^z := Z_s - z_0 + z$.
The order is described by the tail order for each process, $Z$. 
\begin{lem} \label{taillem}
	Let $h$ and $(g,\widetilde g)$ satisfy the same conditions as in
	Theorem~\ref{payoffthm}.
	Let $Z_t^0$ be a Levy procces starting with $0$ such that
	$$
		0 < \min\{  \inf_{0 < u\le\tau} P( Z_u^0>0 ),  \inf_{0<u\le\tau} P( Z_u^0<0 ) \}.
	$$
	Then it holds that as $z \to \infty$,
	\begin{eqnarray*}
		|\hat \epsilon(z)| 
		&=&
		O\left(P( Z_\tau^0 < \log K - z ) \right).
	\end{eqnarray*}
	Also, for any $\varepsilon>0$, it holds that as $z \to -\infty$,
	\begin{eqnarray*}
		|\hat \epsilon(z)| 
		&=&
		O\left(P( Z_\tau^0 > \log K - z )^{1-\varepsilon}\right)
	\end{eqnarray*}
	if $E[ e^{ \frac{1}{\varepsilon} \sup Z_s^0 } ] < \infty$.
\end{lem} 
We next give a description of the order  
when $Z$ is either $\BM$ or $\NIG$; 
\begin{eqnarray}
	&\BM& 
		P( Z_t < \log K - z )
		= O\big( \frac{1}{z} e^{-\frac{ (z+|\log K + \mu\tau|)^2 }{2\sigma^2\tau}} \big) 
		\qquad (z\to\infty)
		\label{BMbeta+}\\
	&\BM& 
		P( Z_t > \log K - z )
		= O\big( \frac{1}{z} e^{-\frac{ (z-|\log K + \mu\tau|)^2 }{2\sigma^2\tau}} \big) 
		\qquad (z\to-\infty)
		\label{BMbeta-}\\
	&\NIG& 
		P( Z_t < \log K - z ) 
		= O(
			|z|^{-\frac{3}{2}}
			e^{ - \theta z - \sqrt{ \theta^2 + \mu^2 } |z| }
		)\qquad(z\to\infty) \label{NIGbeta+}\\
	&\NIG& 
		P( Z_t > \log K - z ) 
		= O(
			|z|^{-\frac{3}{2}}
			e^{ - \theta z - \sqrt{ \theta^2 + \mu^2 } |z| }
		)\qquad(z\to-\infty) \label{NIGbeta-}
\end{eqnarray}
We now prove Theorem~\ref{payoffthm}.
\begin{proof}[Proof of Theorem~\ref{payoffthm}]
In the case of $\BM$, we have the following from the results of Lemma~\ref{taillem}, 
(\ref{BMsup}), (\ref{BMbeta+}), and (\ref{BMbeta-}) that 
\begin{eqnarray*}
	|\epsilon(z)| \sup_{t}  \frac{d\nu_t}{d\bar\nu}  (z) 
	&=& O\left(
		e^{-\frac{ z^2 \rho }{2\sigma^2\tau}} 
    \right)
\end{eqnarray*}
for any $\rho<1$. 
This guarantees the integrability of the integrand in (\ref{integ_mod}).

In the case of $\NIG$, we also have the following from the results of 
Lemma~\ref{taillem}, (\ref{NIGsup}), (\ref{NIGbeta+}), and (\ref{NIGbeta-}) that
\begin{eqnarray*}
	|\epsilon(z)| \sup_{t}  \frac{d\nu_t}{d\bar\nu}  (z) 
	&=&  O(
			|z|^{-\frac{3}{2}\rho}
			e^{ 
				\theta(1-\rho) z
				+ b \frac{ 2 \sqrt{ \theta^2 + \mu^2 } }{ \delta + \sqrt{b^2+\delta^2} } z
				- \rho \sqrt{ \theta^2 + \mu^2 } |z| }
		)
\end{eqnarray*}
for any $\rho<1$. If
$$
	\frac{ 2b }{ \delta + \sqrt{b^2+\delta^2} } < 1,
$$
this guarantees the integrability of the integrand in (\ref{integ_mod}).
\end{proof}
\begin{proof}[Proof of Theorem~\ref{payoffthm2}]
Similarly to Theorem~\ref{payoffthm}, 
the theorem is proved by using a following lemma
instead of Lemma~\ref{taillem}.
\end{proof}
\begin{lem} \label{taillem2}
	Let $(g,\widetilde g)$ satisfy the same conditions as in
	Theorem~\ref{payoffthm2}.
	Let $Z_t^0$ be a Levy procces starting with $0$ such that
	$$
		0 < \min\{  \inf_{0 <  u\le\tau} P( Z_u^0>0 ),  \inf_{0 <  u\le\tau} P( Z_u^0<0 ) \}.
	$$
	Then for $\epsilon > 0$, it holds that as $|z| \to \infty$,
	\begin{eqnarray*}
		|\hat \epsilon(z)| 
		&=&
		O\left(	\min\{
			P( Z_\tau^0 > \log M_U - z )^{1-\varepsilon},
			P( Z_\tau^0 < \log M_L - z )^{1-\varepsilon}
		\} \right)
	\end{eqnarray*}
	if $E[ e^{ \frac{1}{\varepsilon} (\sup Z_s^0 - \inf Z_s^0) } ] < \infty$.
\end{lem}



\subsection{Proof of lemmas}
Note that Lemma~\ref{taillem} and  Lemma~\ref{taillem2} do not require that 
$Z$ be exact for either the $\BM$ or the $\NIG$ case, 
but that
$Z$ has a ``reflection principle'';
\begin{lem} 
	Let $Z_t^0$ be a Levy procces starting with $0$.
	Then for any $a>0$,
	\begin{eqnarray*}
		P(Z_\tau^0>a) &\ge& 
		    D P( \sup_{0\le s\le\tau} Z_s^0>a ) \quad \mbox{and}\\
		P(Z_\tau^0< -a) &\ge& 
			D P( \inf_{0\le s\le\tau} Z_s^0< -a ).
	\end{eqnarray*}
	where
	$$
		D := \min\{  \inf_{0 < u\le\tau} P( Z_u^0>0 ),  \inf_{0< u\le\tau} P( Z_u^0<0 ) \}.
	$$
\end{lem} 
\begin{proof}
	Let $\tau_{a}$ be the hitting time of $(a,\infty)$; i.e.,
	$ \tau_{a} := \inf \{ s \mid Z_s^0 \in (a,\infty) \}$.
	Since the path of $Z^0$ is right continous, 
        $Z_{\tau_{a}}^{0} \ge a$ if $\tau_{a} < \infty$.
	Therefore, the strong Markov property implies
	\begin{eqnarray*}
		P( Z_{\tau}^{0} > a )
		&=&   P( Z_{\tau}^{0} > a, \tau_{a} \le \tau ) \\
		&\ge& P( Z_{\tau}^{0} > Z_{\tau_{a}}^{0}, \tau_{a} \le \tau ) \\
		&=& P( Z_{\tau}^{0} - Z_{\tau_{a}}^{0} > 0, \tau_{a} \le \tau ) \\
		&=&  \int_{0}^{\tau} P( Z_{\tau-s}^{0} > 0 ) P( \tau_{a}=\,ds )\\
		&\ge& \inf_{0< u\le\tau} P( Z_u^0>0 ) \int_{0}^{\tau} P( \tau_{a}=\,ds )\\
		&=& \inf_{0< u\le\tau} P( Z_u^0>0 ) P( \tau_{a} \le \tau ) \\
		&=& \inf_{0< u\le\tau} P( Z_u^0>0 ) P( \sup_{0\le s\le\tau} Z_{s}^{0} > a ).
	\end{eqnarray*}
The proof of the other case is similar.
\end{proof}

We now prove Lemma~\ref{taillem} and Lemma~\ref{taillem2}:
\begin{proof}[Proof of Lemma~\ref{taillem}]
From the definition of $\hat\epsilon(z)$ and the assumption of $(g,\widetilde g)$,
\begin{eqnarray*}
	\hat\epsilon(z) 
	&=&
	E[g( (e^{Z_s^z})_{ 0 \le s \le \tau } ) 
		-\widetilde g( ( e^{Z_s^z} )_{ 0 \le s \le \tau } )] \\
	&=& E[ 
		(h((e^{Z_s^z})_{0 \le s \le\tau})-K)^{+} - C(e^{Z_{0}^z}-K/C)^{+}
	] \\
	&=& E[ (h( ( e^{z+Z_s^0} )_{0\le s \le\tau})-K)^{+}] - C(e^{z}-K/C)^{+}.
\end{eqnarray*}
When $z \le \log\frac{K}{C}$, we have the following inequality 
from (\ref{hbounded}), the H\"older inequality and the reflection principle that  
\begin{eqnarray*}
	\hat\epsilon(z) 
	&=& E[ (h((e^{z+Z_s^0} )_{0\le s \le\tau})-K)^{+}] \\
	&\le&  E[ (\sup_{0\le s \le\tau} e^{z+Z_s^0} -K)^{+}] \\
 	&\le& E[ ( e^z\exp(\sup_{0\le s \le\tau} Z_s^0) - K)^{p} ]^{1/p} 
        P( \sup_{0\le s \le\tau} Z_s^0 > \log K - z)^{1/q} \\
 	&\le& E[ ( K/C\exp(\sup_{0\le s \le\tau} Z_s^0) - K)^{p} ]^{1/p} 
        P( \sup_{0\le s \le\tau} Z_s^0 > \log K - z)^{1/q} \\
	&\le& D^{-1/q} E[ (K/C\exp(\sup_{0\le s \le\tau} Z_s^0) - K)
        ^{p} ]^{1/p} P(Z_\tau^0 > \log K - z)^{1/q},
\end{eqnarray*}
where $1 < q < \infty$ with $\frac{1}{q} + \frac{1}{p} =1$. 
This yields the result
$$
	\hat\epsilon(z)= O\left( P(Z_\tau^0 > \log K - z )^{1/q} \right)
$$
as $z \to -\infty$.

On the other hand, when  $z > \log\frac{K}{C}$,
(\ref{schalar}), (\ref{hbounded}) and the reflection principle imply that
\begin{eqnarray*}
	\hat\epsilon(z) 
		&=& E[ ( h((e^{Z_s^z})_{0\le s \le\tau})-K)^{+}] - (Ce^{z} - K) \\
		&=& E[ ( h((e^{z+Z_s^0})_{0\le s \le\tau})-K)^{+}] - 
                (E[ h((e^{Z_s^0})_{0\le s \le\tau} )]e^{z}-K)\\
        &=& E[ (e^zh((e^{Z_s^0})_{0\le s \le\tau})-K)^{-}] \\
	    &\le& E[ ( e^z \inf_{0\le s \le\tau} e^{Z_s^0} - K)^{-}] \\
	    &\le& K P( \inf_{0\le s \le\tau} Z_s^0 < \log K -z ) \\
	    &\le& \frac{K}{D} P(Z_\tau^0 < \log K -z ).
\end{eqnarray*}
This yields the result $\hat\epsilon(z)= O(P(Z_\tau^0 < \log K - z ))$ as $z \to \infty$. 
This completes the proof.
\end{proof}

\begin{proof}[Proof of Lemma~\ref{taillem2}]
Note that
\begin{eqnarray*}
	\sup e^{Z_s^z} 
	= e^z \sup e^{Z_s^0} 
	= e^z \inf e^{Z_s^0} \frac{\sup e^{Z_s^0}}{\inf e^{Z_s^0}}
	= \inf e^{Z_s^z} e^{ \sup Z_s^0 - \inf Z_s^0 }.
\end{eqnarray*}
Then from the definition of $\hat\epsilon(z)$,
(\ref{payoffthm2a1}),
(\ref{payoffthm2a2}) and
(\ref{payoffthm2a3}),
we have
\begin{eqnarray*}
	\hat\epsilon(z) 
	&=& 
	E[ 
		g( (e^{Z_s^z})_{0 \le s \le\tau})  - \widetilde g( (e^{Z_s^z})_{0 \le s \le\tau}) 
	] \\
	&=& 
	E[ 
		\left( g( (e^{Z_s^z})_{0 \le s \le\tau})  - \widetilde g( (e^{Z_s^z})_{0 \le s \le\tau})  \right)
		1_{ \{ \sup e^{Z_s^z} > M_{U}, M_{L} > \inf e^{Z_s^z} \} }
	] \\
	&\le&	E[ 
		\left( \sup e^{Z_s^z} + M \right)
		1_{ \{ \sup e^{Z_s^z} > M_{U}, M_{L} > \inf e^{Z_s^z} \} }
	] \\
	&\le& E[ 
		\left( M_L e^{ \sup Z_s^0 - \inf Z_s^0 } + M \right)
		1_{ \{ \sup e^{Z_s^z} > M_{U}, M_{L} > \inf e^{Z_s^z} \} }
	].
\end{eqnarray*}
Therefore, the H\"older inequality and the reflection principle imply
\begin{eqnarray*}
	\hat\epsilon(z) 
	&\le& ( M_L E[ e^{ p(\sup Z_s^0 - \inf Z_s^0) } ]^{\frac{1}{p}} + M)	
	P( \sup e^{Z_s^z} > M_{U}, M_{L} > \inf e^{Z_s^z}  )^{\frac{1}{q}} \\
	&=& (  M_L E[ e^{ p(\sup Z_s^0 - \inf Z_s^0) } ]^{\frac{1}{p}} + M)
	P( \sup Z_s^0 > \log M_{U} - z , \log M_{L} - z  > \inf Z_s^0  )^{\frac{1}{q}} \\
	&\le&  ( M_L E[ e^{ p(\sup Z_s^0 - \inf Z_s^0) } ]^{\frac{1}{p}} + M)
	D^{-\frac{1}{q}}
	\min\{ 
		P( Z_\tau^0 > \log M_{U} - z ),
		P( \log M_{L} - z  > Z_\tau^0  )
	\}^{\frac{1}{q}}. 
\end{eqnarray*}
\end{proof}
\section{Appendix} \label{appendixsection}
In this section,
we present some properties of $\BM$ and $\NIG$; 
$\alpha$ and $\hat\nu$ of each process, 
(\ref{BMsup}),
(\ref{NIGsup}),
(\ref{BMbeta+}),
(\ref{BMbeta-}),
(\ref{NIGbeta+}) and
(\ref{NIGbeta-}).
\subsection{Brownian Motion}
Recall that $Z_t^z = z + \sigma W_t + \mu t$ and 
$Z_t = Z_t^{z_0}$. Then the density function of $Z_t^z$ 
with respect to the Lebesgue measure is
$$
  y \mapsto \frac{1}{\sqrt{2\pi\sigma^2t}}e^{-\frac{(y-\mu t-z)^2}{2\sigma^2 t}}.$$

By setting $\alpha:[0,\infty)\to[0,\infty)$ 
as $\alpha(t) = t^{-\frac12} e^{-\frac{\mu^2 t}{2 \sigma^2}}$,
we have 
$$
  \frac{d\hat\nu_t}{d\bar\nu}(z) 
  = \frac{1}{\sqrt{2\pi\sigma^2}}
	e^{-\frac{(z-z_0)^2}{2\sigma^2 t} + \frac{\mu (z-z_0)}{\sigma^2}}.
$$
From the fact that $0 < e^{-\frac{(z-z_0)^2}{2\sigma^2 t}} \le 1$ for all $t$, 
we have
$$
  \sup_t \frac{d\hat\nu_t}{d\bar\nu}(z) \le
	\frac{1}{\sqrt{2\pi\sigma^2}}
	e^{\frac{\mu (z-z_0)}{\sigma^2}}.
$$
This inequality shows (\ref{BMsup}).

For $y>0$, we have
\begin{eqnarray*}
  P( |Z_\tau^0| > y ) 
  = \Phi( \frac{y-\mu\tau}{\sigma\sqrt{\tau}} )
    + \Phi( \frac{y+\mu\tau}{\sigma\sqrt{\tau}} ),
\end{eqnarray*}
where
\begin{eqnarray*}
	\Phi( \xi ) = \int_{\xi}^{\infty} \frac{1}{\sqrt{2\pi}} e^{-\frac{x^2}{2}} \,dx.
\end{eqnarray*}

It holds from L'Hopital's theorem 
that $1 - \Phi(\xi) = \Phi(\xi) \sim  \frac{1}{\sqrt{2\pi}\xi} e^{-\frac{\xi^2}{2}}$. This shows (\ref{BMbeta+}) and  (\ref{BMbeta-}).

%

\subsection{Normal Inverse Gaussian}
Recall that $Z_t^z  = z + W_{\mbox{\rm IG}(t)} + \theta \mbox{\rm IG}(t) + bt$ 
where $\mu \in \mathbb{R}$, 
and we denote the inverse Gaussian subordinator 
$\mbox{\rm IG}(t) = \inf\{s>0;B_s + \mu s > \delta t\}$ where $\delta > 0$.
Then the density function of $Z_t^z$ 
with respect to the Lebesgue measure is 
\begin{eqnarray*}
  y \mapsto 
  \frac{1}{\pi}
  \sqrt{ 
    \frac{\theta^2+\mu^2}
         {\left(\frac{y-z-bt}{\delta t}\right)^2+1 } 
  }
  e^{\mu \delta t + \theta(y-z-bt)}
  K_1\left(
    \delta t 
    \sqrt{(\theta^2+\mu^2)(1+ \left(\frac{y - z - bt}{\delta t}\right)^2)}
  \right)
\end{eqnarray*}
where $K_1$ is the Bessel function of the third kind (See \cite{DA}). 

We use a estimation of the Bessel function from page378 of \cite{AB};
\begin{eqnarray}\label{mBessel1}
  K_1( y ) &\sim& \sqrt{\frac{\pi}{2y}} e^{-y}.
\end{eqnarray}
We set $\alpha(t)$ as
$$
	\alpha(t) 
		=
		\frac{ 1 }{ \sqrt{t} }
		e^{ t\left(
			\mu \delta 
			- \theta b
			- \sqrt{( b^2 + \delta^2 )( \theta^2 + \mu^2 )}
		\right)}.
$$
Then it holds that
\begin{eqnarray*}
\frac{d\hat\nu_t}{d\bar\nu}(x)
	&\to&
		\frac{
			\delta
		}{
			\sqrt{ 2\pi }
		} 
		\left( 
			\frac{ \theta^2 + \mu^2 }
			     { ( b^2 + \delta^2 )^3 } 
		\right)^{\frac{1}{4}}
		e^{ \left(
				\theta
				+
				b
				\sqrt{ \frac{
					\theta^2 + \mu^2
				}{
					b^2 + \delta^2
				}}
			\right)
			( x-z_0 )
		}, \\
		\sup_t \frac{d\hat\nu_t}{d\bar\nu}(x) 
		&\le&
		\frac{
			1
		}{
			\sqrt{ 2\pi }
		} 
		\left( 
			\frac{\theta^2 + \mu^2}{ \delta^2 } 
		\right)^{\frac{1}{4}}
		e^{ 
			\left(
				\theta
				+ b
				\frac{
					2\sqrt{ \theta^2 + \mu^2 } 
				}{
					\delta + \sqrt{b^2+\delta^2}
				}
			\right)
			(x-z_0)
		}.
\end{eqnarray*}
This shows (\ref{BMsup}).

Note that
$
    \delta t 
    \sqrt{(\theta^2+\mu^2)(1+ \left(\frac{z - z_0 - bt}{\delta t}\right)^2)}\to\infty
$
when $|z|\to\infty$.
By using (\ref{mBessel1}) again, 
we obtain that the density function of $Z_\tau^0$ decays in the order as
\begin{eqnarray*}
	f_{Z_\tau^0}(z) \sim |z|^{-\frac{3}{2}}e^{\theta z - \sqrt{\theta^2+\mu^2}|z| }
\end{eqnarray*}
when $|z|\to\infty$.
By using L'Hopital's theorem, we have (\ref{NIGbeta+}) and  (\ref{NIGbeta-}).

\end{document}